\documentclass[11pt]{article}

\usepackage[margin=1in]{geometry}

\usepackage{graphicx}

\usepackage{caption}
\usepackage{subcaption}

\usepackage{xcolor}
\definecolor{navy}{rgb}{0, 0, 0.75}

\usepackage{amsmath}
\usepackage{amsthm}
\usepackage{amssymb}
\usepackage{mathrsfs}
\usepackage{bm}

\usepackage[shortlabels]{enumitem}

\usepackage[colorlinks, allcolors=navy, linktocpage=true]{hyperref}

\theoremstyle{plain}
\newtheorem{theorem}{Theorem}
\newtheorem{corollary}[theorem]{Corollary}
\numberwithin{theorem}{section}

\usepackage[capitalize]{cleveref}

\newcommand{\abs}[1]{\lvert#1\rvert}

\newcommand{\Abs}[1]{{\left\lvert#1\right\rvert}}

\newcommand{\N}{\mathbb{N}}

\newcommand{\cB}{\mathcal{B}}
\newcommand{\cD}{\mathcal{D}}
\newcommand{\cF}{\mathcal{F}}

\newcommand{\cP}{\mathcal{P}}
\newcommand{\cQ}{\mathcal{Q}}

\newcommand{\cX}{\mathcal{X}}
\newcommand{\cY}{\mathcal{Y}}
\DeclareMathOperator*{\E}{\mathbb{E}}

\newcommand{\hcD}{\widehat{\mathcal{D}}}
\newcommand{\tcD}{\widetilde{\mathcal{D}}}
\newcommand{\TV}{d_{\mathrm{TV}}}

\newcommand{\eps}{\varepsilon}

\newcommand\blfootnote[1]{
    \begingroup \renewcommand\thefootnote{}\footnote{#1}
    \addtocounter{footnote}{-1}
    \endgroup
}

\usepackage{soul}
\usepackage{framed}

\title{Supersimulators}
\author{
    Cynthia Dwork \\
    Harvard University \\
    \url{dwork@seas.harvard.edu}
\and
    Pranay Tankala \\
    Harvard University \\
    \url{pranay_tankala@g.harvard.edu}
}
\date{September 22, 2025}

\begin{document}

\maketitle

\begin{framed}
\noindent\textbf{Notice:} This paper has been superseded by \href{https://arxiv.org/abs/2511.03653}{\texttt{arXiv:2511.03653}}, which incorporates all results presented here and more. Please refer to that paper for the most up-to-date version.
\end{framed}

\begin{abstract}
    We prove that every randomized Boolean function admits a \emph{supersimulator}: a randomized polynomial-size circuit whose output on random inputs cannot be efficiently distinguished from reality with constant advantage, even by polynomially larger distinguishers. Our result builds on the landmark \emph{complexity-theoretic regularity lemma} of Trevisan, Tulsiani and Vadhan (2009), which, in contrast, provides a simulator that fools smaller distinguishers. We circumvent lower bounds for the simulator size by letting the distinguisher size bound vary with the target function, while remaining below an absolute upper bound independent of the target function. This dependence on the target function arises naturally from our use of an iteration technique originating in the graph regularity literature.

    The simulators provided by the regularity lemma and recent refinements thereof, known as \emph{multiaccurate} and \emph{multicalibrated} predictors, respectively, as per H\'ebert-Johnson et al. (2018), have previously been shown to have myriad applications in complexity theory, cryptography, learning theory, and beyond. We first show that a recent multicalibration-based characterization of the computational indistinguishability of product distributions actually requires only (calibrated) multiaccuracy. We then show that supersimulators yield an even tighter result in this application domain, closing a complexity gap present in prior versions of the characterization.
\end{abstract}

\blfootnote{This work was supported in part by Simons Foundation Grant 733782 and Cooperative Agreement CB20ADR0160001 with the United States Census Bureau.}

\newpage

\section{Introduction}

Szemer\'{e}di's regularity lemma \cite{szemeredi1975regular} is a cornerstone result in graph theory with far-reaching consequences, including several celebrated results in number theory. Roughly speaking, the lemma states that any large, dense graph---no matter how complex---can be split into a small number of parts between which the graph's edges are distributed \emph{pseudorandomly}.

Although Szemer\'{e}di's regularity lemma spawned many variants and analogues, these versions were still highly specialized (e.g. to cut sizes in graphs, or Fourier uniformity in finite vector spaces). The \emph{complexity-theoretic regularity lemma} \cite{trevisan2009regularity} generalized and abstracted the concept of regularity by considering \emph{indistinguishability} with respect to an arbitrary collection of \emph{distinguisher} functions defined on an arbitrary domain.

In this paper, we prove a strengthened version of the complexity-theoretic regularity lemma. To illustrate our improvement, consider a setting especially relevant to applications in complexity theory and cryptography: distinguishers computable by size-bounded Boolean circuits that receive a uniformly random input from $\{0, 1\}^n$. In this setting, the regularity lemma states that every randomized Boolean function of arbitrary complexity can be \emph{simulated} by a randomized circuit of size $O(s)$ that fools all distinguishers of size at most $s$.

\begin{theorem}[Special Case of Theorem 1.1 of \cite{trevisan2009regularity}]
\label{thm:ttv-for-circuits}
    For all target functions $g : \{0, 1\}^n \to [0, 1]$, size bounds $s \in \N$, and error tolerances $\eps > 0$, there exists a simulator $h : \{0, 1\}^n \to [0, 1]$ of size\footnote{A circuit $c$ computes a real-valued function $h : \{0, 1\}^n \to [0, 1]$ in binary if there exists $m \in \N$ such that $c$ has $n$ input bits, $m$ output bits, and $h(x) = \sum_{i=1}^m  c_i(x) / 2^{i-1}$ for all inputs $x \in \{0, 1\}^n$, where $c_i$ denotes the $i$\textsuperscript{th} output bit of $c$. By the \emph{complexity} or \emph{circuit size} of $h$, we mean the size of the smallest Boolean circuit that computes it.} at most $O(s/\eps^2)$ such that for all distinguishers $A : \{0, 1\}^{n+1} \to \{0, 1\}$ of size at most $s$,
    \[
        \Bigl\lvert{\scalebox{1.2}{$\mathbb{E}$}_{\substack{x \sim \{0, 1\}^n \\ y|x \sim \cB(g(x))}} \big[A(x, y)\bigr] - \scalebox{1.2}{$\mathbb{E}$}_{\substack{x \sim \{0, 1\}^n \\ y|x \sim \cB(h(x))}} \bigl[A(x, y)\bigr]\Bigr\rvert}\le \eps.\footnote{$\cB(p) \in \Delta(\{0, 1\})$ denotes the Bernoulli distribution with parameter $p$.}
    \]
\end{theorem}

\cref{thm:ttv-for-circuits} ensures the existence of a simulator that is at most
a constant factor larger than the distinguishers, but makes no guarantees about the existence of \textit{smaller} simulators.  At least three different arguments in the literature, two in the original paper \cite{trevisan2009regularity} and one in a subsequent work \cite{chen2018complexity}, assert that this complexity gap is inevitable, and the simulator must always be allowed to be larger than the distinguishers that it is asked to fool.

Surprisingly, all three lower bound arguments can be circumvented by relaxing the requirement that the result hold for all size bounds $s$.  By instead allowing the choice of $s$ to depend on the target function, while remaining below an absolute upper bound independent of the target function, it is possible to design simulators that fool families of distinguishers far more powerful than themselves. In this work, we will refer to such an object as a \emph{supersimulator}. As we shall see, this relaxed quantification in which $s$ depends on $g$ will still be useful for applications.

The existence of supersimulators is implicit in the \emph{code-access outcome indistinguishability} construction of \cite{dwork2021outcome}, in which the distinguishers to be fooled have access to the code of the simulator (and are therefore per force larger), but this was not called out by the authors and its significance not explored. Thus our first result, stated next, can be viewed as a rephrasing of Theorem 5.10 of \cite{dwork2021outcome}, unencumbered by questions of the distinguishers' access to the code of the simulator. We shall discuss this connection further in Sections \ref{sec:related} and \ref{sec:construction}.

\begin{theorem}[Supersimulators]
\label{thm:main-succinct}
    For all $g : \{0, 1\}^n \to \{0, 1\}$, $k \in \N$, and $\eps > 0$, there exists a size bound $s \in [n, (n + (\log 1/\eps)^{O(1)})^{k^{1/\eps^2}}]$ and a simulator $h : \{0, 1\}^n \to [0, 1]$ of size at most $s$ such that for all distinguishers $A : \{0, 1\}^{n+1} \to \{0, 1\}$ of size at most $s^{k}$,
    \[
        \Bigl\lvert{\scalebox{1.2}{$\mathbb{E}$}_{\substack{x \sim \{0, 1\}^n \\ y|x \sim \cB(g(x))}} \big[A(x, y)\bigr] - \scalebox{1.2}{$\mathbb{E}$}_{\substack{x \sim \{0, 1\}^n \\ y|x \sim \cB(h(x))}} \bigl[A(x, y)\bigr]\Bigr\rvert}\le \eps.
    \]
\end{theorem}

To understand the statement of \cref{thm:main-succinct}, suppose that $k = 100$, $\eps = 1/10$, and the average-case complexity of $g$ is superpolynomial in $n$. In this case, one cannot hope to approximate $g$ accurately with a function $h$ of complexity $s = \mathrm{poly}(n)$, but \cref{thm:main-succinct} nevertheless guarantees that there exists a simulator of this size fooling all distinguishers of size up to $s^{100}$, well beyond the capabilities of the simulator provided by \cref{thm:ttv-for-circuits}. The particular choice of $s$ may depend on $g$, but it never leaves the interval $[n, n^{10^{300}}]$, which is fixed and independent of $g$. For clarity, we state this concrete version separately, as a corollary:
\begin{corollary}
    For all $g : \{0, 1\}^n \to \{0, 1\}$, $k \in \N$, and $\eps > 0$, there exists $s = \mathrm{poly}(n)$ and $h : \{0, 1\}^n \to [0, 1]$ of size at most $s$ such that for all $A : \{0, 1\}^{n + 1} \to \{0, 1\}$ of size at most $s^{100}$,
    \[
        \Bigl\lvert{\scalebox{1.2}{$\mathbb{E}$}_{\substack{x \sim \{0, 1\}^n \\ y|x \sim \cB(g(x))}} \big[A(x, y)\bigr] - \scalebox{1.2}{$\mathbb{E}$}_{\substack{x \sim \{0, 1\}^n \\ y|x \sim \cB(h(x))}} \bigl[A(x, y)\bigr]\Bigr\rvert}\le \frac{1}{10}.
    \]
\end{corollary}

While we have stated both Theorems \ref{thm:ttv-for-circuits} and \ref{thm:main-succinct} in the setting of size-bounded Boolean circuits on the domain $\cX = \{0, 1\}^n$, both results generalize to arbitrary distinguishers families on an arbitrary domain. We call this setting the \emph{abstract setting}, which we review further in \cref{sec:preliminaries}.

In \cref{sec:construction}, we prove the general version of \cref{thm:main-succinct}. Our proof is based on an iteration technique from the literature on graph regularity, in which strong notions of regularity are achieved by iterating cheaper constructions with a shrinking sequence of error parameters. In graph theory, this technique has been used to achieve state-of-the-art quantitative bounds for applications to graph removal \cite{fox2011new, conlon2012bounds} and to establish relationships between existing notions of regularity \cite{rodl2010regularity}. In our abstract setting, we use it to increase the power of the simulators, allowing them to fool distinguishers whose complexity surpasses their own by any specified \emph{growth function} $G : \N \to \N$. This proof technique also allows us to incorporate a shrinking error tolerance $\eps : \N \to (0, 1)$ that decays as a function of the simulator size.

The final complexity bound on the simulator will scale with an iterated composition of the growth function with itself. This leads, for example, to the doubly exponential upper bound on $s$ in \cref{thm:main-succinct}, which, while large, remains independent of the target function.

\paragraph{Applications} 

\cref{thm:ttv-for-circuits} and its abstract version have been shown to imply many deep results, such as the hardcore lemma \cite{impagliazzo1995hard}, the weak graph regularity lemma \cite{frieze1996regularity, frieze1996approximation}, and the dense model theorem \cite{green2008primes, tao2008primes, reingold2008dense}, along with new characterizations of computational notions of entropy \cite{vadhan2012pseudoentropy, vadhan2013uniform, zheng2014thesis}, and techniques for leakage simulation in cryptography \cite{jetchev2014fake}. More recently, a strengthened version of the complexity-theoretic regularity lemma based on a tool called \emph{multicalibration} \cite{hkrr} has been shown to yield stronger versions of the aforementioned results \cite{casacuberta2024complexity}, along with new results on \emph{omniprediction} in machine learning \cite{omnipredictors} and new results on the computational indistinguishability of product distributions \cite{marcussen2024characterizing}.

These developments raise a natural question: what are the capabilities of supersimulators in these downstream applications? In this paper, we answer this question in one of these domains (the computational indistinguishability of product distributions), which we shall now discuss.

Given $k$ independent samples from an unknown distribution $\cD_b \in \{\cD_0, \cD_1\}$ over $\{0, 1\}^n$, what is the best distinguishing advantage between $\cD_0$ and $\cD_1$ that one could hope to efficiently achieve? With a computationally unbounded distinguisher, the answer is $d_{\mathrm{TV}}(\cD_0^{\otimes k}, \cD_1^{\otimes k})$, which is easily seen to relate to $d_{\mathrm{TV}}(\cD_0, \cD_1)$ via the following inequalities:
\[
    d_{\mathrm{TV}}(\cD_0^{\otimes k}, \cD_1^{\otimes k}) \le 1 - (1 - d_{\mathrm{TV}}(\cD_0, \cD_1))^k \le k \cdot d_{\mathrm{TV}}(\cD_0, \cD_1).
\]
Analogously, let $d_s(\cD_0^{\otimes k}, \cD_1^{\otimes k})$ denote the best distinguishing advantage attainable by a circuit of size at most $s$. Establishing the relationship between $d_s(\cD_0^{\otimes k}, \cD_1^{\otimes k})$ and $d_s(\cD_0, \cD_1)$ turns out to be significantly more involved \cite{halevi2008degradation, geier2022tight}. To facilitate reasoning about the $k$-fold product, it would be ideal if each $\cD_b$ had a computationally indistinguishable proxy distribution $\widetilde{\cD}_b$ such that $d_s(\cD_0, \cD_1) \approx d_{\mathrm{TV}}(\widetilde{\cD}_0, \widetilde{\cD}_1)$. A recent result based on multicalibration achieves almost exactly this, but requires a second size bound $s' = sk/\eps^{O(1)} + (k/\eps)^{O(1)}$ to make the connection bidirectional \cite{marcussen2024characterizing}. Our first result in this setting improves the leading exponent on $\eps$ to $2$ using calibrated multiaccuracy, a weaker but more efficiently achievable condition than multicalibration that more closely resembles the original notion of complexity-theoretic regularity.\footnote{While \cite{marcussen2024characterizing} already report an exponent of $2$ here, it should be larger, owing to the use multicalibration, all known constructions of which incur a $O(s/\eps^c)$ dependence on $s$ for $c > 2$, depending on the precise notion of multicalibration used (e.g. $c=4$ to achieve $\Pr[\max_{f \in \cF} \abs{\E[f(x)(g(x) - h(x))|h(x)]} \le \eps] \ge 1 - \eps$).}

\begin{theorem}
\label{thm:mpv-cma-succinct}
    For all $s, k \in \N$, $\eps > 0$ and $\cD_0, \cD_1$, there exist $\widetilde{\cD}_0, \widetilde{\cD}_1$ such that $d_s(\cD_b, \widetilde{\cD}_b) \le \eps$ and
    \[
        d_{s}(\cD_0^{\otimes k}, \cD_1^{\otimes k}) - k\eps \le d_{\mathrm{TV}}(\widetilde{\cD}_0^{\otimes k}, \widetilde{\cD}_1^{\otimes k}) \le d_{s'}(\cD_0^{\otimes k}, \cD_1^{\otimes k}) + k\eps,
    \]
    where $s' = O(sk/\eps^2) + (k/\eps)^{O(1)}$. One can enforce $\widetilde{\cD}_0 = \cD_0$ with $s' = O(sk/\eps^4) + (k/\eps)^{O(1)}$.
\end{theorem}

Our second result in this setting eliminates the complexity gap entirely using supersimulators. As in \cref{thm:main-succinct}, we now allow ourselves to alter the distinguisher size bound $s$ within a certain interval, with the specific choice depending on the pair $\cD_0, \cD_1$.

\begin{theorem}
\label{thm:mpv-super-succinct}
    For all $s, k \in \N$, $\eps > 0$, $\cD_0, \cD_1$, there exist $s', \widetilde{\cD}_0, \widetilde{\cD}_1$ such that $d_{s'}(\cD_b, \widetilde{\cD}_b) \le \eps$ and
    \[d_{s'}(\cD_0^{\otimes k}, \cD_1^{\otimes k}) - k\eps \le d_{\mathrm{TV}}(\widetilde{\cD}_0^{\otimes k}, \widetilde{\cD}_1^{\otimes k}) \le d_{s'}(\cD_0^{\otimes k}, \cD_1^{\otimes k}) + k\eps,\]
    where $s' \in [s, k^{1/\eps^2}(s + (k/\eps)^{O(1)})]$. One can enforce $\tcD_0 = \cD_0$ with $s' \in [s, k^{1/\eps^4}(s + (k/\eps)^{O(1)})]$.
\end{theorem}

An advantage of \cref{thm:mpv-super-succinct} is that the term $d_{s'}(\cD_0^{\otimes k}, \cD_1^{\otimes k})$ appears on both the left and right side of the chain of inequalities. This yields a tighter characterization in terms of $d_{\mathrm{TV}}(\widetilde{\cD}_0^{\otimes k}, \widetilde{\cD}_1^{\otimes k})$ than in \cref{thm:mpv-cma-succinct}, in which distinct terms $d_{s}(\cD_0^{\otimes k}, \cD_1^{\otimes k}) \le d_{s'}(\cD_0^{\otimes k}, \cD_1^{\otimes k})$, which potentially differ significantly, are used for the left and right side of the chain of inequalities.

\subsection{Summary of Contributions}

\begin{enumerate}[(1)]
    \item We prove a new, powerful variant of the complexity-theoretic regularity lemma, in which the simulator is capable of fooling distinguishers significantly more complex than itself. We call these \emph{supersimulators}. As in \cite{trevisan2009regularity}, our result holds in the abstract setting, where pseudorandomness is measured in terms of an arbitrary collection of distinguisher functions on an arbitrary domain.
    
    Specializing to the setting of Boolean circuits as in \cref{thm:main-succinct}, our result at first seems better than theoretically possible according to known barriers to improving distribution simulators. We bypass these barriers by letting the target function influence the distinguisher size bound, which we choose using an iteration trick from the graph regularity literature.

    As discussed, the existence of supersimulators is implicit in a construction of \cite{dwork2021outcome}, but this was not called out by the authors. Our proof yields a conclusion that is stronger qualitatively, incorporating a shrinking sequence of error parameters rather than a fixed parameter $\eps > 0$.

    \item We tighten the parameters of a recent, multicalibration-based characterization of the computational indistinguishability of product distributions \cite{marcussen2024characterizing}. Our result relies only on the simpler notion of calibrated multiaccuracy (i.e. complexity-theoretic regularity combined with calibration), leading to improved efficiency and a simpler proof.

    \item Also in the setting of product distributions, we show that incorporating the supersimulators that we derived in (1) yields an even tighter variant of the characterization, fully eliminating a complexity gap present in both the original result of \cite{marcussen2024characterizing} and in our version derived in (2). In the present variant, the variable distinguisher size bound from our supersimulator construction leads to a similar change in quantifiers in the characterization.
\end{enumerate}

\subsection{Related Work}
\label{sec:related}

\paragraph{The Regularity Barrier} In the Boolean circuit case, \cref{thm:ttv-for-circuits} asserts that for any size bound $s \in \N$, there exists a simulator $h$ of size $s' = O(s)$ that fools distinguishers of size $s < s'$. The question of the existence of size-$s'$ simulators $h$ that fool distinguishers of size $s  \gg s'$ has already been considered by multiple works \cite{trevisan2009regularity, chen2018complexity}, with negative results.

First, Remark 1.6 of \cite{trevisan2009regularity} constructed two counterexamples ruling out the possibility of any version of \cref{thm:ttv-for-circuits} with $s \ge (ns')^{1 + \Omega(1)}$. One of their examples involves a target function $g : \{0, 1\} \to \{0, 1\}$ of complexity $\tilde{O}(n s')$ sampled from a family of $O(s' \log s')$-wise independent hash functions. Using a Chernoff-like concentration inequality for $k$-wise independence (e.g. Problem 3.8 of \cite{vadhan2012pseudorandomness}), they show that with high probability over the choice of $g$, every function $h$ of complexity at most $s'$ is at most $\frac{1}{10}$-correlated with $g$. Consequently, one cannot hope for $h$ to fool the distinguisher $f = g$ of size $s = \tilde{O}(ns')$.

Crucially, in the aforementioned counterexample, the target function $g$ depends on the simulator size $s'$ (for example, in one standard construction of the hash family, $g$ would be a polynomial whose degree grows with $s'$). The same is true of their other couterexample, which is based on a black-box application of randomness extractors for efficiently sampleable distributions with high min-entropy. Therefore, these counterexamples do not rule out \cref{thm:main-succinct}, which only asserts that for all $g$, there exists at least one ``good'' simulator size $s'$ in a bounded interval.

The subsequent work of \cite{chen2018complexity} provided an even stronger lower bound on the simulator size, under additional assumptions on its structure. Specifically, they argue that any ``black-box'' $(\cF, \eps)$-regular simulator $h$ for $g$ under $\cD$ must make $\Omega(1/\eps^2)$ oracle calls to functions in $\cF$. In the Boolean circuit case, this result may lead one to believe that simulators of size $s'$ can only hope to fool \emph{substantially} smaller distinguishers, and not even those of slightly smaller sizes $s \in [\eps^2 s', s']$. From one point of view, \cref{thm:main-succinct} shows that in the context of Boolean circuits, removing this structural requirement on $h$ impacts the range of attainable simulator and distinguisher sizes.

\paragraph{Outcome Indistinguishability} The outcome indistinguishability (OI) hierarchy \cite{dwork2021outcome} gives a unified view of the complexity-theoretic regularity lemma and several strengthened variants of it.

Simulators satisfying no-access OI, at the lowest level of the hierarchy, are precisely those given by the complexity-theoretic regularity lemma. They are equivalently known as multiaccurate predictors \cite{hkrr, kearns2018gerrymandering} (see also \cite{dwork2023pseudorandomness}). Simulators satisfying sample-access OI, in which the distinguisher receives $h(x)$ as an auxiliary input, are equivalently known as multicalibrated predictors \cite{hkrr}. In oracle-access OI, distinguishers may make a limited number of oracle calls to the simulators.

In code-access OI, the highest level of the hierarchy, each distinguisher receives a string representation of a circuit computing the simulator. Naturally, in this setting, the construction of \cite{dwork2021outcome} allows the complexity of the distinguishers to scale with the length of this representation, which is presented as an auxiliary input. As we shall see in \cref{sec:construction}, this flexibility can be used to show the existence of supersimulators, regardless of considerations of access to the simulator's code.

\paragraph{Calibrated Multiaccuracy} 
The idea to replace the use of multicalibration (a sometimes costly strengthening of the regularity lemma) with calibrated multiaccuracy (a minimal-cost strengthening of the regularity lemma) to achieve quantitative gains in downstream applications has been a theme of multiple recent and concurrent works.

One such work shows that the concept of omniprediction in machine learning can be achieved via calibrated multiaccuracy \cite{gopalan2023lossoi} (see also \cite{dwork2023pseudorandomness}), which had previously shown to be achievable via multicalibration \cite{omnipredictors}. Another work shows that for hardcore set construction, optimal density can be achieved via calibrated multiaccuracy \cite{casacuberta2025global}, which had previously been shown to be achievable via multicalibration \cite{casacuberta2024complexity}, and even earlier with suboptimal density via the original regularity lemma \cite{trevisan2009regularity}. Most recently, \cite{hu2025generalized} shows that calibrated multiaccuracy can be used to extend and improve the efficiency of prior regularity-based characterizations of computational notions of entropy \cite{vadhan2012pseudoentropy, vadhan2013uniform, zheng2014thesis, casacuberta2024complexity}.

\paragraph{Other Related Work}
It was shown in \cite{blasiok2024neural} that minimizing squared loss over the class of neural networks of a given size automatically implies multicalibration, except for a finite and bounded number of ``unlucky'' sizes. The intuition behind this result is related to the intuition for the existence of supersimulators, in the sense that there cannot be too many (substantial) jumps in accuracy as the model/circuit size is varied.

\nocite{feldman2012fourier, de2012chow}

\subsection{Paper Organization} In \cref{sec:preliminaries}, we review the abstract version of the complexity-theoretic regularity lemma and the iteration technique from the graph regularity literature on which we rely. In \cref{sec:construction}, we present and analyze our general supersimulator construction, proving \cref{thm:main-succinct}. In \cref{sec:mpv}, we prove our two results on the computational indistinguishability of product distributions, one using calibrated multiaccuracy (\cref{thm:mpv-cma-succinct}), and the other using supersimulators (\cref{thm:mpv-super-succinct}).

\section{Mathematical Background}
\label{sec:preliminaries}

In this paper, $\cX$ denotes an arbitrary finite set, $\Delta(\cX)$ denotes the set of probability distributions on $\cX$, and $\{\cX \to \cY\}$ denotes the set of functions from $\cX$ to $\cY$.

\paragraph{Complexity-Theoretic Regularity} Given a distribution $\cD \in \Delta(\cX)$, a family of \emph{distinguishers} $\cF \subseteq \{\cX \to [0, 1]\}$, an error parameter $\eps > 0$, a \emph{target} function $g : \cX \to [0, 1]$, and a \emph{simulator} $h : \cX \to [0, 1]$, we say that $h$ is a \emph{$(\cF, \eps)$-regular simulator for $g$ under $\cD$} if for all $f \in \cF$,
\[
    \Bigl\lvert\E_{x \sim \cD}[f(x)(g(x) - h(x))]\Bigr\rvert \le \eps.
\]
Given any $\cD \in \Delta(\cX)$ and $g : \cX \to [0, 1]$, there exists a trivial $(\cF, \eps)$-regular simulator for $g$ under $\cD$, namely $h = g$. The complexity-theoretic regularity lemma guarantees the existence of a much better simulator, whose complexity does not scale with $g$, but rather depends only on $\cF$ and $\eps$. We will state the lemma in terms of the class $\cF_{(s_1, s_2)}$ of functions $h$ with the following property: there exist functions $f_1, \ldots, f_k \in \cF$ with $k \le s_1$ and a Boolean circuit of size at most $s_2$ that computes the output of $h(x)$ in binary given only $f_1(x), \ldots, f_k(x)$ as input.

\begin{theorem}[Complexity-Theoretic Regularity \cite{trevisan2009regularity}]
\label{thm:ttv}
    For all $\cD \in \Delta(\cX)$, $\cF \subseteq \{\cX \to [0, 1]\}$, $g : \cX \to [0, 1]$, and $\eps \in (0, 1)$, there exists an $(\cF, \eps)$-regular simulator in $\cF_{(O(1/\eps^{2}), \, \tilde{O}(1/\eps^{2}))}$.
\end{theorem}

In the introduction, we had focused on the special case where $\cD$ is the uniform distribution on $\cX = \{0, 1\}^n$ and $\cF$ contains all circuits of size at most $s$, but this need not be the case. We had also considered randomized distinguishers $A(x, y)$ that take as a second input a binary label $y$ sampled from either $\cB(g(x))$ or $\cB(h(x))$, but here we consider deterministic distinguishers $f(x)$. Of course, the latter perspective subsumes the former; given $A$, simply set $f_A(x) = \E[A(x, 1) - A(x, 0)]$, where the expectation is taken over the randomness of $A$. In this case, the distinguishing advantage of $A$ is precisely $\abs{\E_{x \sim \cD}[f_A(x)(g(x) - h(x))]}$.

\paragraph{Multicalibration}

A simulator (a.k.a. predictor) $h$ is said to be \emph{calibrated} if \(\E_{x \sim \cD}[g(x) | h(x)] = h(x)\). Although the simulator provided by \cref{thm:ttv} is not necessarily perfectly calibrated, one can ensure that it is \emph{approximately calibrated} at the cost of a small increase in the complexity of $h$. Specifically, following \cite{gopalan2022low}, we say that $h$ is $\gamma$-calibrated if for all functions $w : [0, 1] \to [0, 1]$,
\[
    \Bigl\lvert\E_{x \sim \cD}[w(h(x))(g(x) - h(x)]\Bigr\rvert \le \gamma.
\]
The following result requires only slight modifications to the proof of \cref{thm:ttv}.

\begin{theorem}
\label{thm:ttv-calibrated}
    For all $\cD \in \Delta(\cX)$, $\cF \subseteq \{\cX \to [0, 1]\}$, $g : \cX \to [0, 1]$, and $0 < \gamma \le \eps \le 1$, there exists an $(\cF, \eps)$-regular and $\gamma$-calibrated simulator in $\cF_{(O(1/\eps^{2}), \, \tilde{O}(1/\gamma^3))}$.
\end{theorem}

We have chosen to work with a fairly strong definition of approximate calibration; various relaxations exist in the literature that would lead to an even milder dependence on $\gamma$ in the complexity of $h$ \cite{gopalan2022low}. More stringent than calibration is \emph{multicalibration} \cite{hkrr, kearns2018gerrymandering}, which is essentially a per-level-set regularity requirement on the simulator. We say that $h$ is $(\cF, \eps)$-multicalibrated if
\[
    \Pr\Bigl[\max_{f \in \cF} \Abs{\E\bigl[f(x)(g(x) - h(x)) \,\big|\,h(x)\bigr]} \le \eps\Bigr] \ge 1 - \eps.
\]
To emphasize the relationship between $(\cF, \eps)$-regularity and $(\cF, \eps)$-multicalibration, we will sometimes follow the now-standard convention of referring to the former as \emph{$(\cF, \eps)$-multiaccuracy}, where ``accuracy'' is short for accuracy in expectation.

One can check that $(\cF, \eps)$-multicalibration implies both $(\cF, O(\eps))$-multiaccuracy and, if $\cF$ contains the constant $0$ and $1$ functions, $O(\eps)$-calibration. Moreover, $(\cF, \eps)$-multicalibration can be achieved at the cost of a moderate increase in the complexity of $h$ relative to $\cF$:

\begin{theorem}[\cite{hkrr}]
\label{thm:multicalibration}
    For all $\cD \in \Delta(\cX)$, $\cF \subseteq \{\cX \to [0, 1]\}$, $g : \cX \to [0, 1]$, and $0 < \eps \le 1$, there exists an $(\cF, \eps)$-multicalibrated simulator in $\cF_{(O(1/\eps^{4}), \, \tilde{O}(1/\eps^4))}$.
\end{theorem}

We emphasize that \cref{thm:multicalibration} will not be used directly in any of our proofs; it is nevertheless a useful point of reference. We also remark all of these concepts have been studied in the so-called \emph{multiclass} case, corresponding to $g : \cX \to \Delta([k])$. In this case, which we do not study in this paper, the complexity of known multicalibration constructions are dramatically worse, degrading exponentially with $k$, motivating the study of whether more lightweight regularity notions suffice for downstream applications.

\paragraph{Graph Regularity Iteration} Given a simple graph $G = (V, E)$, let $g = \bm{1}_E: V \times V \to \{0,1\}$ denote its edge indicator function, and let $\cF = \{\bm{1}_{S \times T} : S, T \subseteq V\}$ contain the indicator functions for all rectangles $S \times T \subseteq V \times V$, sometimes called \emph{cuts}. It turns out that famous and well-studied notions of graph regularity, such as Szemer\'{e}di regularity \cite{szemeredi1975regular} and Frieze-Kannan regularity \cite{frieze1996regularity, frieze1996approximation}, are closely related to multicalibration and multiaccuracy with respect to $\cF$, respectively \cite{trevisan2009regularity, skorski2017crypto, dwork2023pseudorandomness}. This connection suggests that techniques developed in the graph regularity literature may have counterparts in our complexity-theoretic setting.

One such technique is iteration, in which one constructs a sequence of regular partitions $\cP_1, \cP_2, \ldots$ of the same graph $G$. In this method, each $\cP_{i+1}$ is a refinement of $\cP_i$, and each $\cP_{i+1}$ is extremely regular relative to the complexity of $\cP_i$ (more formally, $\cP_{i+1}$ achieves regularity with an error parameter $\eps_{\abs{\cP_i}}$ that decays with the number of parts of $\cP_i$). With some care, one can ensure that there exist consecutive partitions $\cP = \cP_{i^\star}$ and $\cQ = \cP_{i^\star + 1}$, where $i^\star \le O(1/\eps^2)$, such that the $\cP$ and $\cQ$ are distance at most $\eps$ in an appropriately defined metric. Depending on the base regularity notion chosen to instantiate this technique, the resulting pair $(\cP, \cQ)$ may have substantially stronger regularity properties than $\cP$ alone, allowing one to derive, for example, Szemer\'{e}di regularity from the weaker notion of Frieze-Kannan regularity, as well as improved bounds for various combinatorial applications. For more detail about this technique in the graph setting, we refer the reader to \cite{zhao2023gtac}, as well as \cite{rodl2010regularity, fox2011new, conlon2012bounds}.

\section{Construction of Supersimulators}
\label{sec:construction}

In this section, we present two different supersimulator constructions in the abstract setting, both of which imply \cref{thm:main-succinct} when specialized to the setting of size-bounded Boolean circuits.

Our constructions are inspired by the iteration technique from the graph regularity literature discussed in the previous section, in which one takes a sequence of increasingly regular vertex partitions. In the graph context, ``increasingly regular'' means shrinking the error parameter $\eps \to 0$ but usually means keeping the distinguisher class (cuts) fixed. In the abstract setting, ``increasingly regular'' could be interpreted as either shrinking $\eps$ or expanding $\cF$, or both. Our first construction expands $\cF$ while keeping $\eps$ fixed. Our second expands $\cF$ and shrinks $\eps$ simultaneously. Although the second is more general than the first, the first has a simpler proof, so we present it separately.

\subsection{First Construction}
\label{sec:construction-1}

We first show how to construct a simulator $h$ that fools distinguishers whose complexity exceeds that of $h$ by an arbitrary, prespecified \emph{growth} function $G$. Recall that we measure the complexity of $h \in \cF_{(s_1, s_2)}$ by two numbers $s_1$ and $s_2$, which count the number of calls to functions in $\cF$ and the number of additional circuit gates, respectively. Accordingly, we consider a valid growth function to be any nondecreasing $G : \N^2 \to \N^2$ under the partial ordering of $\N^2$ in which $s' \ge s$ if both $s'_1 \ge s_1$ and $s'_2 \ge s_2$. (We say a function $F: \N^2 \to (0, 1)$ is \emph{nondecreasing} or \emph{nonincreasing} if $s \le s'$ implies $F(s) \le F(s')$ or $F(s) \ge F(s')$, respectively.)

\begin{theorem}[Supersimulators, Expanding]
\label{thm:supersimulator-basic}
    For all distributions $\cD \in \Delta(\cX)$, distinguisher families $\cF \subseteq \{\cX \to [0, 1]\}$, target functions $g : \cX \to [0, 1]$, error tolerances $\eps \in (0, 1/2)$, and nondecreasing $G : \N^2 \to \N^2$, there exists a size bound $s \in \N^2$ and a simulator $h \in \cF_s$ such that:
    \begin{itemize}
        \item \textbf{(regularity)} $h$ is $\bigl(\cF_{G(s)}, \eps\bigr)$-regular,
        \item \textbf{(complexity)} $s \le S_{\lfloor 1/3\eps^{2} \rfloor}$, where $S_0 = (1, 1)$ and $S_{i+1} = S_i + G(S_i) + \bigl(0, (\log(1/\eps))^{O(1)}\bigr)$.
    \end{itemize}
\end{theorem}

The key difference between \cref{thm:supersimulator-basic} and \cref{thm:ttv} is that the distinguishers may now be more complex than the simulator. Although $h \in \cF_s$, it fools any distinguisher in the class $\cF_{G(s)}$, which can be much larger than $\cF_s$ for appropriately chosen growth functions $G(s) \gg s$. For example, in the Boolean circuit setting, letting $\cF$ consist of the $n$ coordinate functions $x \mapsto x_i$ for $x \in \{0, 1\}^n$ and defining $G(s_1, s_2) = (n, \max(n, s_2)^k)$ yields \cref{thm:main-succinct}. Note also that the complexity of $h$ relative to $\cF$ is still bounded above by a quantity independent of the target function $g$. Indeed, both $s_1$ and $s_2$ can be bounded above by a constant that depends only on the growth function $G$ and the error function $\eps$.

There are several important consequences of being capable of fooling distinguishers sufficiently larger than oneself. For example, one can check that if a simulator $h \in \cF_{s}$ is $(\cF_{s'}, \eps)$-regular, where $s_1' \ge s_1$ and $s_2' \ge s_2 + 1/\eps^{O(1)}$, then $h$, rounded to integer multiples of $\eps$, is automatically $O(\eps)$-calibrated. Roughly speaking, this is because a distinguisher in $\cF_{s'}$ has enough circuit gates to first compute the value $h(x)$ itself and then use its additional $1/\eps^{O(1)}$ gates to evaluate a weighted calibration test (see \cref{sec:preliminaries}). Similarly, if $s_1' \ge s_1 + 1$ and $s_2' \ge s_2 + 1/\eps^{O(1)}$, then the $\eps$-rounded version of $h$ must be $(\cF, \eps^{O(1)})$-multicalibrated. An analogous statement can be made for oracle-access OI, in which distinguishers are allowed to make $q$ oracle calls to $h$, if $s_1'$ and $s_2'$ are at least $q \cdot s_1$ and $q \cdot s_2$ respectively.

Despite these gains, the proof of \cref{thm:supersimulator-basic} is only marginally more cumbersome than that of \cref{thm:ttv}. The boosting-style structure of the proof is the same as that of \cref{thm:ttv}, but one now allows the family of candidate distinguishers to expand dramatically at each step. Interestingly, a modification like this was made in the graph regularity context to give an alternate proof of Szemer\'{e}di's regularity lemma---see Theorem 3.2 of \cite{lovasz2007analyst}, in which the distinguisher family consists of arbitrary unions of increasingly many rectangles. The proof is anlogous to the proof of Lemma 5.6 of \cite{dwork2021outcome} regarding code-access outcome indistinguishability.

\begin{proof}[Proof of \cref{thm:supersimulator-basic}]
    We define functions $h_i : \cX \to [0, 1]$ and $f_i : \cX \to [-1, +1]$ inductively, as follows. First, set $h_0(x) = 1/2$ for all $x \in \cX$. Next, for each $i$, let $\sigma_i \in \{-1, +1\}$ and $f_i \in \sigma_i \cdot \cF_{G(S_i)}$ maximize the correlation $\E_{x \sim \cD}\bigl[f_i(x)(g(x) - h_i(x))\bigr]$. Finally, define $h_{i+1}(x)$ as the projection of $h_i(x) + \eps f_i(x)$ onto the unit interval, rounded to the nearest integer multiple of $\eps' \le \eps^{10}$. Since $\eps'$-precision arithmetic and rounding can be done with $(\log (1/\eps))^{O(1)}$ operations, we have $h_{i + 1} \in \cF_{S_{i+1}}$ for the specified size bound $S_{i+1}$. Next, consider the potential function $\Phi(i) = \E(g(x) - h_i(x))^2$. Clearly, $\Phi(0) \le \frac{1}{4}$ and $0 \le \Phi(i) \le 1$ for all $i$. Furthermore, some algebra shows
    \begin{align*}
        \Phi(i + 1) &\le \E(g(x) - h_i(x) - \eps f_i(x))^2 + 2\eps'\\
        &\le \Phi(i)  - 2\eps \E[f_i(x)(g(x) - h_i(x))] + \eps^2 + 2\eps'.
    \end{align*}
    Phrased differently, since $\eps < \frac{1}{2}$, in each iteration for which $h_i$ is not $(\cF_{G(S_i)}, \eps)$-regular, the potential decreases by at least $\eps^2 - 2\eps' \ge \frac{3}{4}\eps^2$. This decrease can happen at most $\frac{1}{3\eps^2}$ times consecutively starting from $\Phi(0) \le \frac{1}{4}$, so we conclude that $h_i \in \cF_{S_i}$ is $(\cF_{G(S_i)}, \eps)$-regular for some $i \le \frac{1}{3\eps^2}$.
\end{proof}

\subsection{Second Construction}
\label{sec:construction-2}

Our second, more general supersimulator construction, which involves both a growth function and a decaying error function $\eps : \N^2 \to (0, 1)$, more closely resembles iterated graph regularity, both in its statement and in its proof, which is a straightforward iteration of \cref{thm:ttv-calibrated}.

\begin{theorem}[Supersimulators, Expanding/Shrinking]
\label{thm:supersimulator-shrinking}
    For all $\cD \in \Delta(\cX)$, $\cF \subseteq \{\cX \to [0, 1]\}$, $g : \cX \to [0, 1]$, nonincreasing $\eps : \N^2 \to (0, 1/2)$, $\alpha \in (0, 1/2)$, and nondecreasing $G : \N^2 \to \N^2$, there exist $s, s' \in \N^2$ and $h \in \cF_s$ and $h' \in \cF_{s'}$ such that:
    \begin{itemize}
        \item \textbf{(similarlity)} $\E_{x \sim \cD}(h(x) - h'(x))^2 \le \alpha + O(\eps(s))$,
        \item \textbf{(regularity)} $h'$ is $(\cF_{G(s)}, \eps(s))$-regular,
        \item \textbf{(complexity)} $s, s' \le S_{\lfloor 1/\alpha \rfloor}$, where $S_0 = (1, 1)$ and \[S_{i + 1} \le O\bigl(1/\eps(S_i)^{2}\bigr)G(S_i)  + \Bigl(0, \, \tilde{O}\bigl(1/\eps(S_i)^3\bigr)\Bigr).\]
    \end{itemize}
\end{theorem}

\cref{thm:supersimulator-shrinking} provides \emph{two} functions, $h$ and $h'$. It guarantees that that the latter is a simulator that fools distinguishers much more complex than the former. Specifically, while $h \in \cF_s$, the simulator $h'$ fools all distinguishers in $\cF_{G(s)}$, where $G$ is the arbitrary growth function. Moreover, it fools them extremely well, allowing only a vanishingly small distinguishing error of $\eps(s)$, which we may take to decay arbitrarily fast with $s$.

The fact that $h'$ fools distinguishers more complex than $h$ is only useful if we know that $h$ is nontrivial. This is achieved by the \emph{similarity} condition, which states that $h$ and $h'$ are similar to each other in $L^2$ norm. One consequence of this condition is that $h$ itself is a simulator that fools distinguishers larger than itself, albeit with an error parameter that is not vanishingly small:

\begin{corollary}
    Let $\cF, G, \alpha, \eps, h, s$ be as in \cref{thm:supersimulator-shrinking}. Then $h$ is $(\cF_{G(s)}, O(\alpha + \eps(s))^{\frac{1}{3}})$-regular.
\end{corollary}

\begin{proof}
    Fix any $f \in \pm \cF_{G(s)}$. Since $f$ takes values in $[-1, +1]$, for any subset $S \subseteq \cX$,
    \[
        \E\bigl[ f(x)(g(x) - h(x)) \bigr] \le \E\bigl[ f(x)(g(x) - h'(x)) \bigr] + \max_{x \in S}\, \abs{h(x) - h'(x)} + \Pr[x \notin S].
    \]
    It remains to bound the three terms on the right side. For the first term, apply the $(\cF_{G(s)}, \eps(s))$-regularity of $h'$. For the second term, set $S = \{x \in \cX : \abs{h(x) - h'(x)} \le \beta^{\frac{1}{3}}\}$, where $\beta = \alpha + O(\eps(s))$ is the right side of the similarity condition in \cref{thm:supersimulator-shrinking}. For the third term, apply Markov's inequality, which implies that \(\Pr[x \notin S] \le \E (h(x) - h'(x))^2 / \beta^{\frac{2}{3}} = \beta^\frac{1}{3}\).
\end{proof}

As already mentioned, the proof of \cref{thm:supersimulator-shrinking} involves a simple iteration of the calibrated version of the complexity-theoretic regularity lemma (\cref{thm:ttv-calibrated}). We use this calibrated version, rather than the default version (\cref{thm:ttv}), in order to establish the $L^2$ similarity condition. Without the calibration condition, we would only be able to upper bound the signed potential difference between $h$ and $h'$, rather than their $L^2$ distance.

\begin{proof}[Proof of \cref{thm:supersimulator-shrinking}]
    Choose any $h_0 : \cX \to [0, 1]$. For each $i \in \N$, let $h_{i+1}$ be the $(\cF_{G(S_i)}, \eps(S_i))$-regular and $\eps(S_i)$-calibrated predictor that \cref{thm:ttv-calibrated} guarantees lies in
    \[
        (\cF_{G(S_i)})_{\Bigl(O\bigl(1/\eps(S_i)^2\bigr), \, \tilde O\bigl(1/\eps(S_i)^3\bigr)\Bigr)}.
    \] Expanding each call to a function in $\cF_{G(S_i)}$ with calls to functions in $\cF$, we see that $h_{i +1} \in \cF_{S_{i+1}}$ for the specified size bound $S_{i+1}$. Next, define the potential function $\Phi(i) = \E(g(x) - h_i(x)^2)$ as in the proof of \cref{thm:supersimulator-basic}. Similar algebra shows that the $L^2$ distance between $h_i$ and $h_{i+1}$ relates to the difference between $\Phi(i)$ and $\Phi(i+1)$ via:
    \[
         \E(h_i(x) - h_{i+1}(x))^2 = \Phi(i) - \Phi(i+1)- 2\E(h_i(x) - h_{i+1}(x))(g(x) - h_{i+1}(x)).
    \]
    Recall that $(\cF_{G(s)}, \eps(s))$-regularity of $h_{i+1}$ implies \[\Bigl\lvert\E\bigl[h_i(x)(g(x) - h_{i+1}(x))\bigr]\Bigr\rvert \le \eps(s).\] Similarly, $\eps(s)$-calibration of $h_{i+1}$ implies \[\Bigl\lvert\E[h_{i+1}(x)(g(x) - h_{i+1}(x))]\Bigr\rvert \le \eps(s).\] To conclude, choose an $i \le 1/\alpha$ satisfying $\Phi(i) - \Phi(i + 1) \le \alpha$, and let $(h, h') = (h_{i}, h_{i + 1})$.
\end{proof}

\section{Computational Indistinguishability of Product Distributions}
\label{sec:mpv}

In this section, we present our two results on product distributions. In \cref{sec:mpv-cma}, we present our result using calibrated multiaccuracy. In \cref{sec:mpv-super}, we present our result using supersimulators.

\subsection{Characterization via Calibrated Multiaccuracy}
\label{sec:mpv-cma}

In this section, we prove \cref{thm:mpv-cma-succinct}, which quantitatively improves the main result of \cite{marcussen2024characterizing}. Essentially, the improvement comes from replacing multicalibration with calibrated multiaccuracy. The first part of \cref{thm:mpv-cma-succinct}, regarding two proxy distributions $\tcD_0$ and $\tcD_1$, follows from \cref{thm:mpv-cma-1} below---simply let $\cF$ contain all size-$s$ circuits, let $\gamma = \eps^{10}$, and let $h$ be the calibrated simulator provided by \cref{thm:ttv-calibrated}. Making this same substitution in \cref{thm:mpv-cma-2} yields the second part of \cref{thm:mpv-cma-succinct}, regarding a single proxy distribution $\tcD_1$ with $\tcD_0 = \cD_0$.

Curiously, \cref{thm:mpv-cma-1} is a tighter and simpler analysis of essentially the same construction as in \cite{marcussen2024characterizing}, suggesting that their analysis may have only superficially required the stronger assumption of multicalibration. However, for the single-proxy version in \cref{thm:mpv-cma-2}, we must introduce a new construction that was not present in \cite{marcussen2024characterizing}. Indeed, the corresponding construction and analysis in \cite{marcussen2024characterizing} genuinely requires the full strength of multicalibration.

To state our results, recall that $\cD_0, \cD_1 \in \Delta(\cX)$ are $(\cF, \eps)$-indistinguishable if for all $f \in \cF$,
\[
    \Abs{\E_{x \sim \cD_0}[f(x)] - \E_{x \sim \cD_1}[f(x)]} \le \eps.
\]
Given two distributions $\cD_0$ and $\cD_1$, the following theorem finds an $(\cF, \eps)$-indistinguishable proxy for each, which we call $\tcD_0$ and $\tcD_1$, respectively, such that the information-theoretic distinguishability between ($k$-fold products of) the proxy pair roughly matches 
the computational distinguishability between ($k$-fold products of) the original pair.

To build intuition for the construction of $\tcD_0$ and $\tcD_1$, consider the problem of predicting the value of an unknown, uniformly random bit $y \in \{0, 1\}$ given only an observation $x \in \cX$ which is sampled from $\cD_0$ if $y = 0$ or from $\cD_1$ if $y = 1$. An equivalent way to generate the pair $(x, y)$ is to first sample $x$ from the balanced mixture $\frac{1}{2}(\cD_0 + \cD_1)$ and then sample $y$ from its conditional distribution given $x$, namely $\cB(g(x))$ for $g(x) = \frac{\cD_1(x)}{\cD_0(x)+ \cD_1(x)}$. Keeping the marginal distribution of $x$ fixed and replacing the function $g$ with some efficient simulator $h$ yields a pair $(x, \tilde{y})$ with a possibly different distribution. Nevertheless, we can again view the generation of $(x, \tilde{y})$ in two equivalent ways, depending on which of $x$ or $\tilde{y}$ we sample first. Returning to the view in which the bit $\tilde{y} \in \{0, 1\}$ is sampled first, we \emph{define} the proxies $\tcD_0$ and $\tcD_1$ to be the conditional distribution of $x$ given $\tilde{y} = 0$ and $\tilde{y} = 1$, respectively. Ultimately, the proof of \cref{thm:mpv-cma-1} will show that this construction has the desired properties as long as $h$ is calibrated and multiaccurate:

\begin{theorem}
\label{thm:mpv-cma-1}
    Given $\cF \subseteq \{\cX \to [0, 1]\}$, $\cD_0, \cD_1 \in \Delta(\cX)$, $\eps \in (0, 1)$, $\gamma \in (0, \frac{1}{10})$, and $h : \cX \to [0, 1]$, consider the distribution $\cD$ over pairs $(x, y) \in \cX \times \{0, 1\}$ that we define in two equivalent ways:
    \begin{itemize}
        \item[(1a)] sample $y \sim \cB(1/2)$ and then $x|y \sim \cD_y$,
        \item[(1b)] sample $x \sim \cD_\cX$ and then $y | x \sim \cB(g(x))$, for $\cD_\cX =  \frac{1}{2}(\cD_0 + \cD_1)$ and $g = \frac{\cD_1}{\cD_0 + \cD_1}$.
    \end{itemize}
    Consider the following proxy distribution $\tcD$ over pairs $(x, \tilde{y}) \in \cX \times \{0, 1\}$:
    \begin{itemize}
        \item[(2a)] sample $x \sim \cD_\cX$ and then $\tilde{y} | x \sim \cB(h(x))$.
    \end{itemize}
    Then, there exists a unique $p \in [0, 1]$ and $\tcD_0, \tcD_1 \in \Delta(\cX)$ such that the following is equivalent:
    \begin{itemize}
        \item[(2b)] sample $\tilde{y} \sim \cB(p)$ and then $x|y \sim \tcD_y$. 
    \end{itemize}
    Moreover, if $h$ is an $(\cF, \eps)$-regular and $\gamma$-calibrated simulator for $g$ under $\cD_\cX$, then:
    \begin{enumerate}[(i)]
        \item $\cD_b$ and $\tcD_b$ are $(\cF, 2\eps + 5\gamma)$-indistinguishable for each bit $b \in \{0, 1\}$,
        \item $\cD_0^{\otimes k}$ and $\cD_1^{\otimes k}$ are distinguished with advantage $\TV(\tcD_0^{\otimes k}, \tcD_1^{\otimes k}) - 14k\gamma $ by
        \[
            h'(z_1, \ldots, z_k) = \bm{1}{\left[\prod_{i=1}^k h(z_i) > \prod_{i=1}^k (1-h(z_i))\right]}.
        \]
    \end{enumerate}
\end{theorem}

\begin{proof}
    Informally, the idea behind the proof is to establish that
    \[
        \cD_1 = 2g\cD_\cX, \quad \cD_0 = 2(1-g)\cD_\cX, \quad \tcD_1 \approx 2h\cD_\cX, \quad \tcD_0 \approx 2(1-h)\cD_\cX,
    \]
    from which parts (i) and (ii) both follow. First, observe that $\cD_1 = 2g\cD_\cX$ and $\cD_0 = 2(1 - g)\cD_\cX$ are direct consequences of the definitions of $\cD_\cX$ and $g$. To compute similar formulas for the proxy distributions $\tcD_0$ and $\tcD_1$, set $\hcD_0 = 2(1-h)\cD_\cX$ and $\hcD_1 = 2h\cD_\cX$ and observe that for any $z \in \cX$,
    \[
        \tcD_b(z) = \Pr[x = z \,|\,\tilde{y} = b] = \frac{\Pr[\tilde{y} = b \,|\, x = z] \Pr[x = z]}{\Pr[\tilde{y} = b]}.
    \]
    The numerator is $h(z)\cD_\cX(z)$ if $b = 1$ and $(1 - h(z))\cD_\cX(z)$ if $b = 0$. By $\gamma$-calibration,
    \[
        \Pr[\tilde{y} = 1] = \E[h(x)] = \E[g(x)] \pm \gamma = \frac{1}{2} \pm \gamma.
    \]
    It follows that \(1/\Pr[\tilde{y} = b] = 2 \pm 5\gamma\) for all $\gamma < 1/10$, so \(\TV(\tcD_b, \hcD_b) \le 5\gamma\). Next, for $f \in \cF$,
    \[
        \Abs{\E_{x \sim \cD_1}[f(x)] - \E_{x \sim \hcD_1}[f(x)]} = \Abs{\E_{x \sim \cD_\cX}[2f(x)(g(x) - h(x))]} \le 2\eps.
    \]
    Thus, $\cD_1$ and $\hcD_1$ are $(\cF, 2\eps)$-indistinguishable, as are $\cD_0$ and $\hcD_0$. Similarly, by $\gamma$-calibration, $\cD_b$ and $\hcD_b$ are $(\{w \circ h\}, 2\gamma)$-indistinguishable for any $w : [0, 1] \to [0, 1]$. 

    At this point, we have established the statistical closeness of $\tcD_b$ and $\hcD_b$, as well as the computational indistinguishability of $\hcD_b$ and $\cD_b$. When combined, these two properties immediately imply part (i) of the theorem. To prove part (ii), first observe that the function $h'$ defined in the theorem statement optimally distinguishes $\hcD_0^{\otimes k}$ from $\hcD_1^{\otimes k}$ by definition, so
    \[
        \Abs{\E_{x \sim \hcD_0^{\otimes k}}[h'(x)] - \E_{x \sim \hcD_1^{\otimes k}}[h'(x)]} = \TV(\hcD_0^{\otimes k}, \hcD_1^{\otimes k})\ge \TV(\tcD_0^{\otimes k}, \tcD_1^{\otimes k}) - 10k\gamma.
    \]
    To conclude the proof, it remains to replace both expectations over $\hcD_b^{\otimes k}$ on the left with expectations over $\cD_b^{\otimes k}$. We do this with a hybrid argument. In detail, for each $b \in \{0, 1\}$ and $j \in \{0, 1\ldots, k\}$, consider the hybrid distributions $\cD^{(j)}_b = \cD_b^{\otimes j} \otimes \hcD_b^{\otimes (k - j)}$. To relate the expectations of $h'$ under $\cD_b^{(j)}$ and $\cD_b^{(j+1)}$, observe that for any $z_1 \in \cX^{j}$ and $z_3 \in \cX^{k-j - 1}$, there exists a threshold $\tau_{z_1, z_3} \in [0, 1]$ and corresponding threshold indicator function $w_{z_1, z_3}(p) = \bm{1}[p > \tau_{z_1, z_2}]$ such that
    \(
        h'(z_1, z_2, z_3) = w_{z_1, z_3}(h(z_2))
    \)
    for all $z_2 \in \cX$. Thus, the $(\{w_{z_1, z_3} \circ h\}, 2\gamma)$-indistinguishability of $\cD_b$ and $\hcD_b$ implies
    \[
        \Abs{\E_{x \sim \cD_b^{(j)}}[h'(x)] - \E_{x \sim \cD_b^{(j+1)}}[h'(x)]}
        \le \E_{\substack{x_{1} \sim \cD_b^{\otimes j} \\ x_{3} \sim \hcD_b^{\otimes (k-j-1)}}}\Abs{\Pr_{x_2 \sim \cD_b}[h(x_2) > \tau_{x_1, x_3}] - \Pr_{x_2 \sim \hcD_b}[h(x_2) > \tau_{x_1, x_3}]} \le 2\gamma.
    \]
    To conclude the proof, note that $\cD_b^{(0)} = \hcD_b^{\otimes k}$ and $\cD_b^{(k)} = \cD_b^{\otimes k}$ and apply the triangle inequality.
\end{proof}

We now turn our attention to the single-proxy version of our result (i.e. the second part of \cref{thm:mpv-cma-succinct}), in which we enforce the constraint that $\tcD_0 = \cD_0$. Note that $\tcD_1$ still may differ from $\cD_1$. The construction of $\tcD_1$ in this setting will follow the same pattern as in the preceding proof: we first consider a pair $(x, y)$ in which $x|y \sim \cD_y$, replace the conditional distribution of $y|x$ with a calibrated and multiaccurate simulator to obtain a pair $(x, \tilde{y})$, and finally define $\tcD_1$ to be the conditional distribution of $x|(\tilde{y}=1)$. The key difference will be in the choice of the marginal distribution of $y$. In the preceding proof, we took $y \sim \cB(1/2)$. In the following proof, we will take $y \sim \cB(\eps)$ for a small parameter $\eps > 0$. Roughly speaking, this change ensures that the marginal distribution of $x$ is significantly tilted toward $\cD_0$, eliminating the need for an explicit proxy distribution $\tcD_0$. Because of the many similarities between the preceding theorem statement and the one below, we have written the key differences in blue:

\begin{theorem}
\label{thm:mpv-cma-2}
    Given $\cF$, $\cD_0$, $\cD_1$, $\eps$, and $h$ as in \cref{thm:mpv-cma-1}, let $\gamma \in (0, {\color{navy} \frac{\eps}{2}})$ and suppose we:
    \begin{itemize}
        \item[(1a)] sample $y \sim {\color{navy} \cB(\eps)}$ and then $x|y \sim \cD_y$; or, equivalently,
        \item[(1b)] sample $x \sim \cD_\cX$ and then $y|x \sim \cB(g(x))$, for $\cD_\cX = {\color{navy} (1-\eps)\cD_0 + \eps\cD_1}$ and $g = {\color{navy} \frac{\eps \cD_1}{(1-\eps)\cD_0 + \eps\cD_1}}$.
    \end{itemize}
    Define $\tcD_1$ as in \cref{thm:mpv-cma-1}. If $h$ is $(\cF, {\color{navy} \eps^2})$-regular and $\gamma$-calibrated for $g$ under $\cD_\cX$, then:
    \begin{enumerate}[(i)]
        \item $\cD_1$ and $\tcD_1$ are $(\cF, {\color{navy} \eps + \frac{2\gamma}{\eps^2}})$-indistinguishable,
        \item $\cD_0^{\otimes k}$ from $\cD_1^{\otimes k}$ are distinguished with advantage at least $\TV(\cD_0^{\otimes k}, \tcD_1^{\otimes k}) - {\color{navy} \bigl(\frac{2\gamma}{\eps^2} + \frac{\gamma}{\eps} + \eps\bigr)k}$ by
        \[
            h'(z_1, \ldots, z_k) = \bm{1}{\left[\prod_{i=1}^k h(z_i) > {\color{navy} \eps^k}\right]}.
        \]
    \end{enumerate}
\end{theorem}
\begin{proof}
    Informally, the key observation is that $\cD_\cX$ is so tilted toward $\cD_0$ that we have no need for the proxy $\tcD_0$. More formally, we first observe that $\cD_1 = g\cD_\cX / \eps$. To compute a similar formula for the proxy distribution $\tcD_1$, we set $\hcD_1 = h\cD_\cX/\eps$ and observe that for any $z \in \cX$,
    \[
        \tcD_1(z) = \Pr[x = z | \tilde{y} = 1] = \frac{\Pr[\tilde{y} = 1 \,|\, x = z] \Pr[x = z]}{\Pr[\tilde{y} = 1]}.
    \]
    The numerator is $h(z)\cD_\cX(z)$. By $\gamma$-calibration,
    \[
        \Pr[\tilde{y} = 1] = \E[h(x)] = \E[g(x)] \pm \gamma = \eps \pm \gamma.
    \]
    It follows that $1/ \Pr[\tilde{y} = 1] = \frac{1}{\eps} \pm \frac{2\gamma}{\eps^2}$ for all $\gamma < \frac{\eps}{2}$, so $\TV(\tcD_1, \hcD_1) \le \frac{2\gamma}{\eps^2}$. Next, for $f \in \cF$,
    \[
        \Abs{\E_{x \sim \cD_1}[f(x)] - \E_{x \sim \hcD_1}[f(x)]} = \Abs{\E_{x \sim \cD_\cX}\Bigl[\frac{1}{\eps}f(x)(g(x) - h(x))\Bigr]} \le \eps.
    \]
    Thus, $\cD_1$ and $\hcD_1$ are $(\cF, \eps)$-indistinguishable. Similarly, by $\gamma$-calibration, $\cD_1$ and $\hcD_1$ are $(\{w \circ h\}, \frac{\gamma}{\eps})$-indistinguishable for any $w : [0, 1] \to [0, 1]$.
    
    At this point, we have shown that $\tcD_1$ and $\hcD_1$ are statistically close and that $\hcD_1$ and $\cD_1$ are computationally indistinguishable. Together, these two properties immediately imply part (i). To prove part (ii), observe that by definition, the function $h'$ defined in the theorem statement optimally distinguishes $\cD_\cX^{\otimes k}$ from $\hcD_1^{\otimes k}$, so
    \[
        \Abs{\E_{x \sim \cD_\cX^{\otimes k}}[h'(x)] - \E_{x \sim \hcD_1^{\otimes k}}[h'(x)]} = \TV(\cD_\cX^{\otimes k}, \hcD_1^{\otimes k})\ge \TV(\cD_\cX^{\otimes k}, \tcD_1^{\otimes k}) - \frac{2\gamma}{\eps^2}k.
    \]
    To conclude the proof, it remains to replace $\cD_\cX^{\otimes k}$ on both the left and right sides with $\cD_0^{\otimes k}$, and also to replace the expectation over $\hcD_1^{\otimes k}$ on the left side with $\cD_1^{\otimes k}$. Since $\TV(\cD_\cX^{\otimes k}, \cD_0^{\otimes k}) \le k \eps$, the first of these changes can be made at the cost of an additive $2k\eps$ penalty. By the same hybrid argument as in the proof of \cref{thm:mpv-cma-1}, the second incurs an additional $\frac{\gamma}{\eps}k$.
\end{proof}

\subsection{Characterization via Supersimulators}
\label{sec:mpv-super}

In this section, we prove \cref{thm:mpv-super-succinct}, which fully closes the complexity gap between the upper and lower bounds by combining the supersimulator of \cref{thm:supersimulator-basic} with Theorems \ref{thm:mpv-cma-1} and \ref{thm:mpv-cma-2}.

\begin{proof}[Proof of \cref{thm:mpv-super-succinct}]
    First, observe that in \cref{thm:mpv-cma-1}, whenever $h$ has Boolean circuit complexity at most $s'$, the corresponding $h'$ has complexity at most $G(s') = ks' + (k/\eps)^{O(1)}$. Thus, given any $s \in \N$, it suffices to apply \cref{thm:mpv-cma-1} to an $\eps^{10}$-calibrated simulator $h$ of complexity $s'$ that fools distinguishers of complexity $G(\max(s, s'))$ with error $\eps$. Applying \cref{thm:supersimulator-basic} with an appropriately chosen growth function shows that there exists such a simulator $h$ with complexity $s' \le k^{1/\eps^2}(s + (k/\eps)^{O(1)})$. Finally, observe that starting from \cref{thm:mpv-cma-2} instead of \cref{thm:mpv-cma-1} doubles the leading exponent on $\eps$ from $2$ to $4$ due to its $\eps^2$-regularity requirement.
\end{proof}

\bibliographystyle{alpha}
\bibliography{sources}

\end{document}